\documentclass[copyright]{eptcs}
\providecommand{\event}{J.~Dassow, G.~Pighizzini, B.~Truthe (Eds.): 11th International Workshop\\
on Descriptional Complexity of Formal Systems (DCFS 2009)} % Name of the event you are submitting to
\providecommand{\volume}{3}
\providecommand{\anno}{2009}
\providecommand{\firstpage}{127} %127--136
\providecommand{\eid}{12}
\usepackage{breakurl}        % Not needed if you use pdflatex only.

\input{dcfs.tex}

\usepackage{graphicx}
\usepackage{latexsym}

\title{State Complexity Approximation}
\def\titlerunning{State Complexity Approximation}

\author{\makebox[0pt][c]{Yuan Gao}\hspace{4cm}\makebox[0pt][c]{Sheng Yu}
\institute{Department of Computer Science --
University of Western Ontario\\
London N6A 5B7 -- Ontario -- Canada}
\email{\makebox[0pt][c]{ygao72@csd.uwo.ca}\hspace{4cm}\makebox[0pt][c]{syu@csd.uwo.ca}}
}
\def\authorrunning{Y.~Gao, S.~Yu}

\begin{document}

\maketitle

\begin{abstract}
In this paper, we introduce the new concept of state complexity
approximation, which is a further development of state complexity
estimation. We show that this new concept is useful in both of the
following two cases: the exact state complexities are not known
and the state complexities have been obtained but are in
incomprehensible form. 
\end{abstract}

\newtheorem{apptheo}{T\/heorem}[section]
\newtheorem{conjecture}{Conjecture}[section]
\newtheorem{applem}{Lemma}[section]
\newcommand{\propersubset}{\subset}
\newtheorem{open}{Open problem}[section]

\section{Introduction}

The state complexity of combined operations has been studied in,
e.\,g., \cite{com1,com2,EGLY}. It has been shown that the state
complexity of combined operations is at least as important and
practical as the state complexity of individual operations. There
is only a limited number of individual operations on regular
languages. However, the number of combined operations on regular
languages is unlimited and each of them is not simply a
mathematical composition of the state complexities of their
component individual operations. It appears that the exact state
complexity of each combined operation has to be studied
specifically.

There are at least the following two problems concerning the state
complexities for combined operations. First, the state
complexities of many combined operations are extremely difficult
to compute. Second, a large proportion of results that have been
obtained are pretty complex and impossible to comprehend. For
example, the state complexity of the catenation for four regular
languages accepted by $m, n, p, q$ states, respectively, is
$$9(2m-1)2^{n+p+q-5}-3(m-1)2^{p+q-2}-(2m-1)2^{n+q-2}+(m-1)2^q+(2m-1)2^{n-2}.$$

It is clear that close estimations of state complexities are good
enough in many automata applications. In \cite{kai,EGLY},
estimations of state complexity of combined operations have been
proposed and studied. In this paper, we go further in the
direction of the study in \cite{kai,EGLY} and introduce the
concept of state complexity approximation. Briefly speaking, an
approximation of a state complexity is an estimate of the state
complexity with a ratio bound clearly defined. The ratio bound
gives a precise measurement on the quality of the estimate.

The idea of state complexity approximation is from the notion of
approximation algorithms which was formalized in early 1970's by
David S. Johnson et al. \cite{approx1,approx2,approx3}. Many
polynomial-time approximation algorithms have been designed for a
quite large number of NP-complete problems, which include the
well-known travelling-salesman problem, the set-covering problem,
and the subset-sum problem. Obtaining an optimal solution for an
NP-complete problem is considered intractable. Near optimal
solutions are often good enough in practice. Assuming that the
problem is a maximization or a minimization problem, an
approximation algorithm is said to have a ratio bound of $\rho(n)$
if for any input of size $n$, the cost~$C$ of the solution
produced by the algorithm is within a factor of $\rho(n)$ of the
cost~$C^*$ of an optimal solution \cite{algorithm}:
$$\max\left(\frac{C}{C^*}, \frac{C^*}{C}\right) \leq \rho(n).$$

The concept of state complexity approximation is in many ways
similar to that of approximation algorithms. A state complexity
approximation is close to the exact state complexity and normally
not equal to it. The ratio bound shows the error range of the
approximation. In addition to the property of having a small ratio
bound in general, we also consider that a state complexity
approximation should be in a simple and intuitive form.

In spite of the similarities, there are fundamental differences
between a state complexity approximation and an approximation
algorithm. The efforts in the area of approximation algorithms are
in finding polynomial algorithms for NP-complete problems such
that the results of the algorithms approximate the optimal
results. In comparison, the efforts in the state complexity
approximation are in searching directly for the estimations of
state complexities such that they satisfy certain ratio bounds.
The aim of designing an approximation algorithm is to transform an
intractable problem into one that is easier to compute and the
result is acceptable although not optimal. In comparison, a state
complexity approximation result may have two different effects:
(1) it gives a reasonable estimation of certain state complexity,
with some bound, the exact value of which is difficult or
impossible to compute; or (2) it gives a simpler and more
comprehensible formula that approximates a known state complexity.

In the next section, we give some basic definitions and notation
including the formal definition of state complexity approximation.
In Section 3, we show the state complexity approximation results
on four basic combined operations: the star of union, the star of
intersection, the star of catenation, and the star of reversal. In
Section 4, we show that state complexity approximation results can
be easily obtained for some operations the exact state
complexities of which may be very difficult to obtain. In Section
5, we show that certain state complexity can be very complex in
formulation. A state complexity approximation is clearly more
intuitive and comprehensible than the exact state complexity. In
Section~6, we conclude the paper.

\section{Preliminaries}

A deterministic finite automaton (DFA) is denoted by a 5-tuple $A
\!=\! (Q, \Sigma, \delta, s, F),$ where $Q$ is the finite and nonempty
set of states, $\Sigma$ is the finite and nonempty set of input
symbols, $\delta: Q\times\Sigma \rightarrow Q$ is the state
transition function, $s\in Q$ is the initial state, and
$F\subseteq Q$ is the set of final states. A DFA is said to be
complete if $\delta$ is a total function. Complete DFAs are the
basic model for considering state complexity. Without specific
mentioning, all DFAs are assumed to be complete in this paper.

A nondeterministic finite automaton (NFA) is also denoted by a
5-tuple \hbox{$M = (Q, \Sigma, \delta, s, F)$}, where $Q$, $\Sigma$, $s$,
and $F$ are defined the same way as in a DFA and $\delta:
Q\times\Sigma\rightarrow 2^Q$ maps a pair consisting of a state
and an input symbol into a set of states rather than, more
restrictively, a single state. An NFA may have multiple initial
states, in which case an NFA is denoted $(Q, \Sigma, \delta, S,
F)$ where~$S$ is the set of initial states.

The reader may refer to \cite{Hopcroft,Salomaa,Yu} for a rather
complete background knowledge in automata theory.

State complexity (\cite{Yu}) is a descriptional complexity measure
for regular languages based on the deterministic finite automaton
model. So, by state complexity we mean the deterministic state
complexity.

The {\em state complexity} of a regular language $L$, denoted
$sc(L)$, is the number of states in the minimal complete DFA
accepting $L$. When we speak about the state complexity of a
(combined) operation on regular languages, we mean the worst case
state complexity of the languages resulting from the operation as
a function of the state complexity of the regular operand
languages.
%Add
So, without specific mentioning, by state complexity we mean the
worst-case state complexity in the following.
%end addition

If the above definition is based on minimal NFA rather than
minimal complete DFA, we have the nondeterministic state
complexity, which has been studied in~\cite{HK,markus2}.

Let $\xi$ be a combined operation on $k$ regular languages. Assume
that the state complexity of $\xi$ is $\theta$. We say that
$\alpha$ is a state complexity approximation of the operation
$\xi$ with the ratio bound $\rho$ if, for any large enough
positive integers $n_1, \ldots, n_k$, which are the numbers of
states of the DFAs that accept the argument languages of the
operation, respectively, $$\max\left(\frac{\alpha(n_1, \ldots,
n_k)}{\theta(n_1, \ldots, n_k)}, \frac{\theta(n_1, \ldots,
n_k)}{\alpha(n_1, \ldots, n_k)}\right) \leq \rho(n_1, \ldots,
n_k).$$ Note that in many cases, $\rho$ is a constant. Since state
complexity is a worst-case complexity, an approximation that is
not smaller than the actual state complexity is preferred, which
is the case for every approximation result in this paper.

\section{Some basic results on state complexity approximation}

In \cite{kai}, an estimation method through nondeterministic state
complexities was introduced for the (deterministic) state
complexities of certain types of combined operations. The method
is described in the following.

Assume we are considering the combination of a language operation
$g_1$ with~$k$ arguments together with operations  $g_2^i$, $i =
1, \ldots, k$. The {\em nondeterministic estimation upper bound,}
or {\em NEU-bound\/} for the deterministic state complexity of the
combined operation $g_1(g_2^1, \ldots, g_2^k)$ is calculated as
follows:

\begin{itemize}
\item[{\rm (i)}]
Let the arguments of the operation $g^i_2$ be DFAs $A^i_j$ with
$m^i_j$ states,  \hbox{$i = 1, \ldots, k$}, $j = 1, \ldots, r_i$, $r_i
\geq 1$.
\item[{\rm (ii)}] The nondeterministic
state complexity of the combined operation is at most the
composition of the individual state complexities, and hence the
language
$$
g_1( g_2^1(L(A^1_1), \ldots, L(A^1_{r_1})), \ldots,
g_2^k(L(A^k_1), \ldots, L(A^k_{r_k})))
$$
has an NFA with at most
$$
{\rm nsc}(g_1)( {\rm nsc}(g^1_2)(m^1_1, \ldots, m^1_{r_1}),
\ldots, {\rm nsc}(g^k_2)(m^k_1, \ldots, m^k_{r_k}))
$$
states, where {\rm nsc}$(g)$ is the nondeterministic state
complexity (as a function) of the language operation $g$.
\item[{\rm (iii)}] Consequently,
the deterministic state complexity of the combined operation
$g_1(g_2^1, \ldots, g_2^k)$ is upper bounded by
\begin{equation}
\label{NEU} 2^{ {\rm nsc}(g_1)( {\rm nsc}(g^1_2)(m^1_1, \ldots,
m^1_{r_1}), \ldots, {\rm nsc}(g^k_2)(m^k_1, \ldots, m^k_{r_k})) }
\end{equation}
\end{itemize}

The nondeterministic state complexity of the basic individual
operations on regular languages has been investigated in
\cite{HK,markus2,ellul}.

In the following we show that this estimation method can produce
nice approximation results for the state complexities of certain
combined operations. The table below shows the actual state
complexities and their corresponding NEU-bounds of the four
combined operations \cite{kai}: (1) star of union, (2) star of
intersection, (3) star of catenation, and (4) star of reversal.

{\small\renewcommand{\arraystretch}{1.2}
\begin{center}
\begin{tabular}{|c|c|c|c|}
\hline
Operations & State Complexity & NEU-bound \\
\hline $(L(A)\cup L(B))^*$ & $2^{m+n-1}-2^{m-1}-2^{n-1}+1$
& $2^{m+n+2}$\\
\hline $(L(A)\cap L(B))^*$ & $ 3/4\; 2^{mn}$
& $2^{mn+1}$ \\
\hline $(L(A)L(B))^*$ & $2^{m+n-1}+2^{m+n-4}-2^{m-1}-2^{n-1}+m+1$
&
$2^{m+n+1}$ \\
\hline
$(L(B)^R)^*$ & $2^n$ & $2^{n+2}$\\
\hline
\end{tabular}
\end{center}
}

The next table shows clearly that each NEU-bound in the previous
table gives a very good approximation to its corresponding state
complexity.

{\small\renewcommand{\arraystretch}{1.2}
\begin{center}
\begin{tabular}{|c|c|}
\hline
Operations & Ratio bounds of the approximation \\
\hline
$(L(A)\cup L(B))^*$ & $\approx 8$ \\
\hline
$(L(A)\cap L(B))^*$ & $8/3$ \\
\hline
$(L(A)L(B))^*$ & $4$ \\
\hline
$(L(B)^R)^*$ & $4$ \\
\hline
\end{tabular}
\end{center}
}

In the above cases, although the exact state complexities have
been obtained, the approximation results with small ratio bounds
are good enough for practical purposes, and they clearly have the
advantage of being more intuitive and simpler in formulation.

\section{Approximation without knowing actual state complexity}

In this section, we consider two combined operations: (1) star of
left quotient and (2) left quotient of star. For each of the
combined operations, we do not have the exact state complexity;
however, an approximation with a good ratio bound is obtained.

Let $R$ and $L$ be two languages over the alphabet $\Sigma$. Then
the left quotient of~$R$ by $L$, denoted $L\backslash R$, is the
language
$$\{y\mid xy\in R \mbox{ and } x\in L\}.$$

In the following, we assume that all languages are over an
alphabet of at least two letters.

\subsection{The state complexity approximation of star of left quotient}

%% The state complexity of left quotient of an $n$-state regular
%% language by an arbitrary language is $2^n-1$ and the state
%% complexity of star of an $n$-state regular language is
%% $2^{n-1}+2^{n-2}$(\cite{YZS}). A direct composition of the two
%% functions would be doubly exponential $2^{2^n-2}+2^{2^n-3}$. It is
%% obviously an upper bound of the state complexity of the combined
%% operation star of left quotient, however, it cannot be reached.
%% Below we compute an approximation of the state complexity of this
%% combined operation which is more precise.

\begin{theorem}
\label{Theorem1}

Let $R$ be a language accepted by an $n$-state DFA $M$, $n>0$, and
$L$ be an arbitrary language. Then there exists a DFA of at most
$2^n$ states that accepts $(L\backslash R)^*$.
\end{theorem}

\begin{proof} Let $M=(Q,\Sigma , \delta , s, F)$ be a complete DFA of $n$
states and $R=L(M)$. For each $q\in Q$, denote by $L(M_q)$ the set
$\{ w\in \Sigma ^*|\delta (s,w)=q \}$. We construct an NFA $M'$
with multiple initial states to accept $(L\backslash R)^+$ as
follows. $M'$ is the same as $M$ except that the initial state $s$
of $M$ is replaced by the set of initial states $S=\{ q|L(M_q)\cap
L \neq \emptyset \}$ and $\varepsilon$-transitions are added from
each final state to the states in $S$. By using subset
construction, we can construct a DFA $A'$ of no more than $2^n-1$
states that is equivalent to $M'$. Note that $\emptyset$ is not a
state of $A'$. From the DFA $A'$, we construct a new DFA $A$ by
just adding a new initial state that is also a final state and the
transitions from this new state that are the same as the
transitions from the original initial state of $A'$. It is easy to
see that $L(A)=(L\backslash R)^*$ and $A$ has $2^n$ states.
%% So, star of left quotient will not produce a DFA
%% of more than $2^n$ states. We can consider $2^n$ an approximation
%% of the state complexity of star of left quotient.
\end{proof}

This result gives an upper bound for the state complexity of the
combined operation: star of left quotient. %It means that for any
%$n$-state DFA language $R$, $n>0$, and an arbitrary language $L$,
%the state complexity of the star of the left quotient of $R$ by
%$L$ is no more than $2^n$.

\begin{theorem}
\label{Theorem2} For any integer $n\geq 2$, there exist a DFA $M$
of $n$ states and a language $L$ such that any DFA accepting
$(L\backslash L(M))^*$ needs at least $2^{n-1}+2^{n-2}$ states.
\end{theorem}
\begin{proof} For $n=2$, it is clear that $R=\{ w\in
\{a,b\}^*|\#_a(w)\mbox{ is odd}\}$ is accepted by a two-state DFA,
and $$(\{ \varepsilon\}\backslash R)^*=R^*=\{ \varepsilon\}\cup \{
w\in \{a,b\}^*|\#_a(w)\geq 1\}$$ cannot be accepted by a DFA with
less than three states.

For $n>2$, let $M=(Q,\{ a, b\} , \delta , 0, \{n-1 \})$ where $Q =
\{0, 1, \ldots , n-1\}$, $\delta(i, a) = i+1$  mod $n$, $i=0, 1,
\ldots , n-1$, $\delta(0, a) = 0$ and $\delta(j, b) =j+1$ mod $n$,
$j=1, \ldots , n-1.$

%For $n>2$, let $M=(Q,\{ a, b\} , \delta , 0, \{n-1 \})$ where $Q
%= \{0, 1, \ldots , n-1\}$ and
%\begin{eqnarray*}
%\delta(i, a) & = & i+1  \ \ {\rm{ mod }} \ \  n, \ \ i=0, 1, \ldots , n-1,\\
%\delta(0, a) & = & 0,\\
%\delta(j, b) & = & j+1  \ \ {\rm{ mod }} \ \  n, \ \ j=1, \ldots ,
%n-1.
%\end{eqnarray*}
%The transition function of $M$ is shown in Figure~\ref{Star}.
%\begin{figure}[ht]
%  \centering
%  \includegraphics[scale=0.75]{Star_1.eps}
%  \caption{The transition diagram of DFA $M$}
%\label{Star}
%\end{figure}
It has been proved in~\cite{YZS} that the minimal DFA accepting
$L(M)^*$ has \hbox{$2^{n-1}+2^{n-2}$} states. Let $L=\{ \varepsilon \}$.
Then $(L\backslash L(M))^*=L(M)^*$. So, any DFA accepting
$(L\backslash L(M))^*$ needs at least $2^{n-1}+2^{n-2}$ states.
\end{proof}

This result gives a lower bound for the state complexity of star
of left quotient. Clearly, the lower bound does not coincide with
the upper bound. We still do not know the exact state complexity
for this combined operation, yet, which could be difficult to
obtain. However, we can easily obtain a good state complexity
approximation for the operation. Let $2^n$ the approximation. Then
the ratio bound would be
$$\frac{2^{n}}{2^{n-1}+2^{n-2}}= \frac{4}{3}\ \ .$$

\subsection{The state complexity approximation of left quotient of star}

%% A direct composition of the state complexity of star and that of
%% left quotient is $2^{2^{n-1}+2^(n-2)}-1$. This doubly exponential
%% upper bound of the state complexity of the combined operation left
%% quotient of star is also too high to be reached. In the following
%% we show an approximation that is more close to the exact state
%% complexity of this combined operation.

Here we consider the combined operation: left quotient of star.

\begin{theorem}
\label{Theorem5} Let $R$ be a language accepted by an $n$-state
DFA $M$ and $L$ an arbitrary language. Then there exists a DFA of
at most $2^{n+1}-1$ states that accepts $L\backslash R^*$.
\end{theorem}

\begin{proof} Let $M\!\!=\!\!(Q,\Sigma , \delta , s, F)$ be a complete DFA of $n$
states and $R\!\!=\!\!L(M)$. Then we can easily construct an
$(n+1)$-state NFA \hbox{$M'=(Q\cup \{s'\} ,\Sigma, \delta', s',
F\cup\{s'\})$} such that $L(M')=R^*$ by adding a new initial state
$s'$ and transitions \hbox{$\delta'(s',\varepsilon)=s$} and
$\delta'(f,\varepsilon)=s'$ for each final state $f\in F$. For
each $q\in Q\cup \{s' \}$, we denote by $L(M_q)$ the set $\{ w\in
\Sigma ^*|q \in \delta'(s',w) \}$. We construct an NFA $N$ with
multiple initial states to accept $L\backslash L(M')=L\backslash
R^*$ as follows. $N$ is the same as $M'$ except that the initial
state $s'$ of $M'$ is replaced by the set of initial states 
$S=\{q \mid L(M_p)\cap L \neq \emptyset \}$. By using subset construction,
we can verify that there exists a DFA $A$ of no more than
$2^{n+1}-1$ states that is equivalent to $N$. Note that $\emptyset
$ is not a state of $A$. It is easy to see that
$$L(A)=L(N)=L\backslash L(M')=L\backslash R^*.$$ So, $2^{n+1}-1$
is an upper bound of the state complexity of left quotient of
star.\end{proof}

\begin{theorem}
\label{Theorem6} For any integer $n\geq 2$, there exist a DFA $M$
of $n$ states and a language $L$ such that any DFA accepting
$L\backslash L(M)^*$ needs at least $2^{n-1}+2^{n-2}$ states.
\end{theorem}
\begin{proof} For $n=2$, we still use $R=\{ w\in\{a,b\}^*|\#_a(w)\mbox{
is odd}\}$ which is accepted by a two-state DFA. $\{
\varepsilon\}\backslash R^*=R^*$ cannot be accepted by a DFA with
less than three states.

%Again we use DFA $M$ shown in Figure~\ref{Star} for any integer
%$n>2$. As stated before, it has been proved that the minimal DFA
%accepting $L(M)^*$ has $2^{n-1}+2^{n-2}$ states. So any DFA
%accepting $L\backslash L(M)^*$ needs at least $2^{n-1}+2^{n-2}$
%states.

Again we use the same DFA $M$ defined in the proof of
Theorem~\ref{Theorem2} for any integer $n>2$. As stated before, it
has been proved that the minimal DFA accepting $L(M)^*$ has
$2^{n-1}+2^{n-2}$ states. So any DFA accepting $L\backslash
L(M)^*$ needs at least $2^{n-1}+2^{n-2}$ states.
\end{proof}

For this combined operation, we choose $2^{n+1}$ to be an
approximation of its state complexity. Then the ratio bound can be
calculated easily as follows:
$$\frac{2^{n+1}}{2^{n-1}+2^{n-2}} = \frac{8}{3}\ \ .$$

\section{State complexity approximation of the catenation of regular languages}

As we know, the state complexity of the catenation of an
$n_1$-state DFA language and an $n_2$-state DFA language, $n_1\geq
1$ and $n_2\geq 2$, is $n_1 2^{n_2} - 2^{n_2-1}$ (\cite{YZS}). The
state complexity of multiple catenations has been studied in
~\cite{EGLY} and the following estimate was obtained.

\begin{claim}
\label{general-estimate} Let $R_1, \ldots, R_k$, $k\geq 2$, be
regular languages accepted by DFAs of $n_1, \ldots, n_k$ states,
respectively. Then the state complexity of $R_1\cdots R_k$ is no
more than
$$n_12^{n_2+\cdots +n_k}-2^{n_2+\cdots +n_k-1}-2^{n_3+\cdots
+n_k-1}-\cdots -2^{n_k-1}.$$
\end{claim}

%% This claim clearly holds when $k=2$. For $k>2$, the above bound is
%% obtained by induction using the formula $T(k)\leq
%% T(k-1)2^{n_k}-2^{n_k-1}$.

The exact state complexity of the catenations of three and four
regular languages was also obtained in ~\cite{EGLY}. In this
section, we prove the exact state complexities of the catenation
of $k$ regular languages for arbitrary $k\geq 2$. Note that this
is not a state complexity in the normal definition that is for
only one specific (combined) operation. This is a state complexity
(formula) for a class of (combined) operations.
%% It is also pretty close to the estimate above.

After we prove this state complexity, we show an approximation of
the complexity and state why the approximation is useful in this
case.

We first consider a lower bound.

\begin{theorem}
\label{lowerbound}

For any integers $n_i\geq 2$, $1\leq i\leq k$, there exist DFA
$A_i$ of $n_i$ states, respectively, such that any DFA accepting
$L(A_1)\cdots L(A_k)$ needs at least
\begin{eqnarray*}
 & & n_12^{n_2+\dots +n_k}-D-\sum_{i = 1}^{k-1} E_i
\end{eqnarray*}
states, where
\begin{eqnarray*}
D&=&n_1(2^{n_3+\dots +n_k}-1)+n_1(2^{n_2}-1)(2^{n_4+\dots
+n_k}-1)+\dots +n_1(2^{n_2}-1)\cdots(2^{n_{k-2}}-1)(2^{n_k}-1);\\
E_1&=& 1+(2^{n_2-1}-1)(1+(2^{n_3}-1)(1+(2^{n_4}-1)\cdots
(1+(2^{n_{k-1}}-1)2^{n_k})\ldots ));\\
E_2&=&(n_1-1)2^{n_2-1}(1+(2^{n_3-1}-1)(1+(2^{n_4}-1)\cdots
(1+(2^{n_{k-1}}-1)2^{n_k})\ldots ))\\
& &{}+2^{n_2-2}(1+(2^{n_3-1}-1)(1+(2^{n_4}-1)\cdots
(1+(2^{n_{k-1}}-1)2^{n_k})\ldots ));\\
%E_3&=&(n_1-1)(2^{n_2-1}-1)2^{n_3-1}(1+(2^{n_4-1}-1)(1+(2^{n_5}-1)\cdots
%(1+(2^{n_{k-1}}-1)2^{n_k})\ldots ))\\
%& &+(n_1-1)2^{n_2-1}2^{n_3-2}(1+(2^{n_4-1}-1)(1+(2^{n_5}-1)\cdots
%(1+(2^{n_{k-1}}-1)2^{n_k})\ldots ))\\
%& &+2^{n_2-2}2^{n_3-1}(1+(2^{n_4-1}-1)(1+(2^{n_5}-1)\cdots
%(1+(2^{n_{k-1}}-1)2^{n_k})\ldots ))\\
%& &+2^{n_2-2}2^{n_3-2}(1+(2^{n_4-1}-1)(1+(2^{n_5}-1)\cdots
%(1+(2^{n_{k-1}}-1)2^{n_k})\ldots ));\\
\ldots\\
%%%%%%%%%%%%%%%%%%%%%%%%%%%%%%E_i has been added here%%%%%%%%%%%%%%%%%%%%%%%%%%%%%%
E_i&=&(n_1-1)(2^{n_2-1}-1)\cdots 2^{n_i-1}(1+(2^{n_{i+1}-1}-1)(1+(2^{n_{i+2}}-1)\cdots 
(1+(2^{n_{k-1}}-1)2^{n_k})\ldots ))\\
& &{}+\dots +2^{n_2-2}\cdots 2^{n_i-2}(1+(2^{n_{i+1}-1}-1)(1+(2^{n_{i+2}}-1)\cdots 
(1+(2^{n_{k-1}}-1)2^{n_k})\ldots )).%;\\
%\ldots\\
%E_{k-1}&=&(n_1-1)(2^{n_2-1}-1)(2^{n_3-1}-1)\cdots
%(2^{n_{k-2}-1}-1)2^{n_{k-1}-1}2^{n_k-1}\\
%& & +\ldots +2^{n_2-2}2^{n_3-2}\cdots
%2^{n_{k-2}-2}2^{n_{k-1}-2}2^{n_k-1}.
\end{eqnarray*}
\end{theorem}
\begin{proof} Let $\Sigma=\{a_j\mid 1\leq j \leq 2k-1 \}$. Define a DFA
$A_1=(Q_1,\Sigma , \delta _1, 0, F_1)$, where 
\begin{align*}
Q_1 &= \{ 0, 1, \ldots , n_1\!-\!1\},\\
F_1 &= \{ n_1-1\},\\
\delta _1 (t,a_{1}) &= t+1 \mod n_1,\ 0\leq t\leq n_1-1, \\
\delta _1 (t,a_{2k-2}) &= 0,\ 0\leq t\leq n_1-1,\\ 
\delta _1 (t,b) &= t, \ b\in \Sigma-\{a_{1}, a_{2k-2}\},\ 0\leq t\leq n_1-1.
\end{align*}
%\begin{eqnarray*}
%Q_1 & = & \{ 0, 1, \ldots , n_1-1\};\\
%F_1 & = & \{ n_1-1\};\\
%\delta _1 (t,a_{1}) & = & t+1 \ \ {\rm{ mod }} \ \ n_1, \ \
%0\leq t\leq n_1-1;\\
%\delta _1 (t,a_{2k-2}) & = & 0, \ \ 0\leq t\leq n_1-1;\\
%\delta _1 (t,b) & = & t, \ \ b\in \Sigma-\{a_{1}, a_{2k-2}\}, \ \
%0\leq t\leq n_1-1.
%\end{eqnarray*}
%Figure~\ref{Ck1} shows the transition diagram of $A_1$.
%\begin{figure}[ht]
%  \centering
%  \includegraphics[scale=0.75]{Ck1_1.eps}
%  \caption{The transition diagram of DFA $A_1$}
%\label{Ck1}
%\end{figure}

Let DFA $A_i=(Q_i,\Sigma , \delta _i, 0, F_i)$, $2\leq i \leq k$,
where 
\begin{align*}
Q_i &= \{ 0, 1, \ldots , n_i-1\},\\
F_i &= \{ n_i-1\},\\
\delta _i (t,a_{2i-2}) &= t+1 \mod n_i,\ 0\leq t\leq n_i-1,\\
\delta _i (t,a_{2i-1}) &= 1,\ 0\leq t\leq n_i-1,\\
\delta _i (t,b) &= t, \ b\in \Sigma-\{a_{2i-2}, a_{2i-1}\}, \ 0\leq t\leq n_i-1.
\end{align*}
%\begin{eqnarray*}
%Q_i & = & \{ 0, 1, \ldots , n_i-1\};\\
%F_i & = & \{ n_i-1\};\\
%\delta _i (t,a_{2i-2}) & = & t+1 \ \ {\rm{ mod }} \ \ n_i, \ \
%0\leq t\leq n_i-1;\\
%\delta _i (t,a_{2i-1}) & = & 1, \ \ 0\leq t\leq n_i-1;\\
%\delta _i (t,b) & = & t, \ \ b\in \Sigma-\{a_{2i-2}, a_{2i-1}\}, \
%\ 0\leq t\leq n_i-1.
%\end{eqnarray*}
%Figure~\ref{Cki} shows the transition diagram of $A_i$.
%\begin{figure}[ht]
%  \centering
%  \includegraphics[scale=0.75]{Ckismall_1.eps}
%  \caption{The transition diagram of DFA $A_i$}
%\label{Cki}
%\end{figure}

For each $x\in \{a_1, a_2, a_4, \ldots, a_{2k-2} \}^*$ and $2\leq
s \leq k$, we define
$$
P_{s}(x)  =  \{ p\mid x=u_1u_2\ldots u_s, \ u_l\in L(A_l), 
1\leq l\leq s-1, \ and \ p=\#_{a_{2s-2}}(u_s) \mod n_s \}.
$$
%$2\leq s \leq k$.

Consider that $x,y\in \{a_1, a_2, a_4, \ldots, a_{2k-2} \}^*$ such
that $P_s(x)\neq P_s(y)$. Let $c\in P_s(x)-P_s(y)$ (or
$P_s(y)-P_s(x)$) and
$w=a_{2s-2}^{n_s-1-c}a_{2s+1}a_{2s}^{n_{s+1}-1}\cdots
a_{2k-1}a_{2k-2}^{n_{k}-1}$. Then it is clear that $xw\in
L(A_1)\cdots L(A_k)$ but $yw \notin L(A_1)\cdots L(A_k)$. So, $x$
and~$y$ are in different equivalence classes of the
right-invariant relation induced by $L(A_1)\cdots L(A_k)$.

For each $x\in \{a_1, a_2, a_4, \ldots, a_{2k-2} \}^*$, define
\begin{eqnarray*}
P_1(x) & = &  \#_{a_1}(z) \text{ where $x=ydz$, $y\in \{a_1, a_2, a_4, \ldots, a_{2k-2} \}^*$},\\
       &   &  \qquad z\in \{a_1, a_2, a_4, \ldots, a_{2k-4} \}^*, \text{ if $a_{2k-2}$ occurs in $x$};\\
P_1(x) & = & \#_{a_1}(x), \text{ otherwise}.
\end{eqnarray*}

Consider $u,v\in \{a_1, a_2, a_4, \ldots, a_{2k-2} \}^*$ such that
\hbox{$P_1(u) \mod n_1 > P_1(v) \mod n_1$}.\\
Let $i=P_1(u)$ mod $n_1$
and $w=a_1^{n_1-1-i}a_3a_{2}^{n_2-1}\cdots
a_{2k-1}a_{2k-2}^{n_{k}-1}$. Then clearly $uw\in L(A_1)\cdots
L(A_k)$ but $vw\notin L(A_1)\cdots L(A_k)$.

Notice that there does not exist a word $w$ such that $0\notin
P_{2}(w)$ and $P_1(w)=n_1-1$, since\linebreak \hbox{$P_1(w)=n_1-1$} guarantees that
$0\in P_{2}(w)$. Because of the same reason, there does not exist
a word~$w$ such that $n_{t}-1\in P_{t}(w)$ and $0 \notin
P_{t+1}(w)$, $2\leq t \leq k-1$. It is also impossible that
$P_{t}(w)=\emptyset$ but $P_{t+1}(w)\neq \emptyset$.

For each subset $p_s=\{d_{1,s},\ldots , d_{e_s,s}\}$ of
$\{0,\ldots ,n_s-1\}$ where $d_{1,s}>\cdots >d_{e_s,s}$ and $2\leq
s\leq k$, and an integer $p_1\in \{0,\ldots , n_1-1\}$, except the
cases we mentioned above, there exists a word
\begin{eqnarray*}
x & = & a_{1}^{n_1}a_{2}^{n_2}a_{4}^{n_3}\cdots a_{2k-4}^{n_{k-1}}a_{2k-2}^{d_{1,k}-d_{2,k}}a_{1}^{n_1}a_{2}^{n_2}a_{4}^{n_3}\cdots a_{2k-4}^{n_{k-1}}a_{2k-2}^{d_{2,k}-d_{3,k}}\cdots \\
 & & a_{1}^{n_1}a_{2}^{n_2}a_{4}^{n_3}\cdots a_{2k-4}^{n_{k-1}}a_{2k-2}^{d_{e_k-1,k}-d_{e_k,k}}a_{1}^{n_1}a_{2}^{n_2}a_{4}^{n_3}\cdots a_{2k-4}^{n_{k-1}}a_{2k-2}^{d_{e_k,k}} \\
 & & a_{1}^{n_1}a_{2}^{n_2}a_{4}^{n_3}\cdots a_{2k-4}^{d_{1,k-1}-d_{2,k-1}}\cdots a_{1}^{n_1}a_{2}^{n_2}a_{4}^{n_3}\cdots a_{2k-4}^{d_{e_{k-1},k-1}}\cdots \\
 & & a_{1}^{n_1}a_{2}^{d_{1,2}-d_{2,2}}\cdots a_{1}^{n_1}a_{2}^{d_{e_2,2}}a_{1}^{p_1}.
\end{eqnarray*}
such that $P_1(x)=p_1$ and $P_s(x)=p_s$.

In total, there are $n_12^{n_2}2^{n_3}\cdots 2^{n_k}$ classes.
There are
\begin{eqnarray*}
D & = & n_1(2^{n_3+\dots +n_k}-1)+n_1(2^{n_2}-1)(2^{n_4+\dots
+n_k}-1)+\dots +n_1(2^{n_2}-1)\cdots (2^{n_{k-2}}-1)(2^{n_k}-1)
\end{eqnarray*}
classes with both $p_{t}=\emptyset$ and $p_{t+1}\neq \emptyset$,
$2\leq t \leq k-1$. There are
\begin{eqnarray*}
E_1=(1+(2^{n_2-1}-1)(1+(2^{n_3}-1)(1+(2^{n_4}-1)\cdots
(1+(2^{n_{k-1}}-1)2^{n_k})\ldots ))
\end{eqnarray*}
classes with both $p_1=n_1-1$ and $0\notin p_2$. There are
\begin{eqnarray*}
E_2&=& (n_1-1)2^{n_2-1}(1+(2^{n_3-1}-1)(1+(2^{n_4}-1)\cdots
(1+(2^{n_{k-1}-1})2^{n_k})\ldots
 ))\\
   & & {}+2^{n_2-2}(1+(2^{n_3-1}-1)(1+(2^{n_4}-1)\cdots (1+(2^{n_{k-1}}-1)2^{n_k})\ldots
 ))
\end{eqnarray*}
classes with both $n_2-1\in p_2$ and $0\notin p_3$, which are not
in $E_1$.
%There are
%\begin{eqnarray*}
% E_3  &=& (n_1-1)(2^{n_2-1}-1)2^{n_3-1}(1+(2^{n_4-1}-1)(1+(2^{n_5}-1)\cdots (1+(2^{n_{k-1}}-1)2^{n_k})\ldots
% ))\\
%     & & +(n_1-1)2^{n_2-1}2^{n_3-2}(1+(2^{n_4-1}-1)(1+(2^{n_5}-1)\cdots (1+(2^{n_{k-1}}-1)2^{n_k})\ldots
% ))\\
%      & & +2^{n_2-2}2^{n_3-1}(1+(2^{n_4-1}-1)(1+(2^{n_5-1})\cdots (1+(2^{n_{k-1}}-1)2^{n_k})\ldots
% ))\\
%      & & +2^{n_2-2}2^{n_3-2}(1+(2^{n_4-1}-1)(1+(2^{n_5-1})\cdots (1+(2^{n_{k-1}}-1)2^{n_k})\ldots
% ))
%\end{eqnarray*}
%classes with both $n_3-1\in p_3$ and $0\notin p_4$, which are not
%in $E_1, E_2$.
We omit the other similar classes until the $i$th group of
classes. There are
%%%%%%%%%%%%%%%%%%%%%%%%%%%%%%E_i has been added here%%%%%%%%%%%%%%%%%%%%%%%%%%%%%%
\begin{eqnarray*}
E_i&=&(n_1-1)(2^{n_2-1}-1)\cdots 2^{n_i-1}(1+(2^{n_{i+1}-1}-1)(1+(2^{n_{i+2}}-1)\cdots (1+(2^{n_{k-1}}-1)2^{n_k})\ldots ))\\
& &{}+\dots +2^{n_2-2}\cdots 2^{n_i-2}(1+(2^{n_{i+1}-1}-1)(1+(2^{n_{i+2}}-1)\cdots (1+(2^{n_{k-1}}-1)2^{n_k})\ldots ))
\end{eqnarray*}
classes with both $n_i-1\in p_i$ and $0\notin p_{i+1}$, which are
not in $E_1, E_2, \ldots , E_{i-1}$.

%We omit the other similar classes until the last group of classes.
%There are
%\begin{eqnarray*}
%        E_{k-1}&=& (n_1-1)(2^{n_2-1}-1)(2^{n_3-1}-1)\cdots (2^{n_{k-2}-1}-1)2^{n_{k-1}-1}2^{n_k-1}
%\\
%          & & +\ldots  +2^{n_2-2}2^{n_3-2}\cdots 2^{n_{k-2}-2}2^{n_{k-1}-2}2^{n_k-1}
%\end{eqnarray*}
%classes with both $n_{k-1}-1\in p_{k-1}$ and $0\notin p_k$, which
%are not in $E_1, E_2,\ldots , E_{k-2}$.

Thus, there are at least 
$n_12^{n_2+\ldots +n_k}-D-E_1-E_2-\dots-E_{k-1}$ distinct equivalence classes. \end{proof}

%Thus, there are at least
%\begin{eqnarray*}
% & & n_12^{n_2+\ldots +n_k}-D-\sum_{i = 1}^{k-1} E_i
%\end{eqnarray*}
%distinct equivalence classes.
%\end{proof}

\begin{theorem}
Let $A_i$, $1\leq i\leq k$ be $k$ DFAs of $n_i$, respectively,
where $A_i$ has $f_i$ final states, $0<f_i <n_i$. Then there
exists a DFA of
\begin{eqnarray*}
 & & n_12^{n_2+\dots +n_k}-D-\sum_{i = 1}^{k-1} E_i
\end{eqnarray*}
states that accepts $L(A_1)\cdots L(A_k)$, where{\footnotesize
\begin{eqnarray*}
D&=&n_1(2^{n_3+\dots +n_k}-1)+n_1(2^{n_2}-1)(2^{n_4+\dots
+n_k}-1)+\dots +n_1(2^{n_2}-1)\cdots (2^{n_{k-2}}-1)(2^{n_k}-1);\\
E_1&=& f_1(1+(2^{n_2-1}-1)(1+(2^{n_3}-1)(1+(2^{n_4}-1)\cdots
(1+(2^{n_{k-1}}-1)2^{n_k})\ldots ));\\
E_2&=&(n_1-f_1)(2^{f_2}-1)2^{n_2-f_2}(1+(2^{n_3-1}-1)(1+(2^{n_4}-1)\cdots
(1+(2^{n_{k-1}}-1)2^{n_k})\ldots ))\\
& &+f_1(2^{f_2}-1)2^{n_2-f_2-1}(1+(2^{n_3-1}-1)(1+(2^{n_4}-1)
\cdots
(1+(2^{n_{k-1}}-1)2^{n_k})\ldots ));\\
%E_3&=&(n_1-f_1)(2^{n_2-f_2}-1)(2^{f_3}-1)2^{n_3-f_3}(1+(2^{n_4-1}-1)(1+(2^{n_5}-1)\cdots \\
%& &(1+(2^{n_{k-1}}-1)2^{n_k})\ldots ))\\
%& &
%+(n_1-f_1)(2^{f_2}-1)2^{n_2-f_2}(2^{f_3}-1)2^{n_3-f_3-1}(1+(2^{n_4-1}-1)
%(1+(2^{n_5}-1)\cdots \\
%& &(1+(2^{n_{k-1}}-1)2^{n_k})\ldots))\\
%& &
%+f_12^{n_2-f_2-1}(2^{f_3}-1)2^{n_3-f_3}(1+(2^{n_4-1}-1)(1+(2^{n_5}-1)\cdots (1+(2^{n_{k-1}}-1)2^{n_k})\ldots ))\\
%& &
%+f_1(2^{f_2}-1)2^{n_2-f_2-1}(2^{f_3}-1)2^{n_3-f_3-1}(1+(2^{n_4-1}-1)(1+(2^{n_5}-1)\cdots \\
%& &(1+(2^{n_{k-1}}-1)2^{n_k})\ldots ));\\
\ldots\\
%%%%%%%%%%%%%%%%%%%%%%%%%%%%%%E_i has been added here%%%%%%%%%%%%%%%%%%%%%%%%%%%%%%
E_i&=&(n_1-f_1)(2^{n_2-f_2}-1)\cdots(2^{f_i}-1)2^{n_i-f_i}(1+(2^{n_{i+1}-1}-1)(1+(2^{n_{i+2}}-1)\cdots 
\quad(1+(2^{n_{k-1}}-1)2^{n_k})\ldots ))+\cdots\\
& &{}+f_1(2^{f_2}-1)2^{n_2-f_2-1}\cdots(2^{f_i}-1)2^{n_i-f_i-1}(1+(2^{n_{i+1}-1}-1)(1+(2^{n_{i+2}}-1)\cdots 
\quad(1+(2^{n_{k-1}}-1)2^{n_k})\ldots )).%;\\
%\ldots \\
%E_{k-1}&=&(n_1-f_1)(2^{n_2-f_2}-1)(2^{n_3-f_3}-1)\cdots (2^{n_{k-2}-f_{k-2}}-1)(2^{f_{k-1}}-1)2^{n_{k-1}-f_{k-1}}2^{n_k-1}\\
%& & +\ldots\\
%& & +f_1(2^{f_2}-1)2^{n_2-f_2-1}(2^{f_3}-1)2^{n_3-f_3-1}\cdots
%(2^{f_{k-2}}-1)2^{n_{k-2}-f_{k-2}-1}(2^{f_{k-1}}-1)\\
%& & 2^{n_{k-1}-f_{k-1}-1}2^{n_k-1}.
\end{eqnarray*}}
\end{theorem}
\begin{proof} Let DFA $A_i=(Q_i,\Sigma , \delta _i, 0, F_i)$, $1\leq i \leq k$. 

Construct $E=(Q_E,\Sigma, \delta_E,q_0,F_E)$ such that
\begin{eqnarray*}
Q_E & = & Q_{1}\times 2^{Q_{2}} \times 2^{Q_{3}} \times \cdots
\times 2^{Q_{k}}-D'-\sum_{i = 1}^{k-1} E_i';\\
q_0 & = & \begin{cases} \langle 0,\emptyset ,\ldots , \emptyset \rangle , & \text{if $0\notin F_i$, $1\leq i\leq k$};\\
              \langle 0,\{0\} ,\ldots ,  \emptyset \rangle , & \text{if $0 \in F_1$ and $0\notin F_i$, $2\leq i\leq k$};\\
              \ldots &\\
              \langle 0,\{0\} ,\ldots ,  \{0\} \rangle , & \text{if $0 \in F_i$, $1\leq i\leq k-1$};
          \end{cases}\\
%q_0 & = & \langle 0,\emptyset ,\ldots , \emptyset \rangle , \ \ {\rm{if}} \ \ 0\notin F_i,  \ \  1\leq i\leq k;\\
%& = & \langle 0,\{0\} ,\ldots ,  \emptyset \rangle , \ \ {\rm{if}} \ \ 0 \in F_1, \ \ {\rm{and}} \ \ 0\notin F_i,  \ \  2\leq i\leq k;\\
%& & \ldots \\
%& = & \langle 0,\{0\} ,\ldots ,  \{0\} \rangle , \ \ {\rm{if}} \ \ 0 \in F_i, \ \  1\leq i\leq k-1;\\
F_E & = & \{ \langle u_1,u_2, \ldots ,u_k \rangle \in Q_E \mid u_k\cap F_k\neq \emptyset\};\\
\delta_E & : & \delta_E( \langle u_1, u_2, \ldots ,u_k \rangle ,a)
= \langle u_1',u_2', \ldots ,u_k' \rangle , \text{ for $a\in \Sigma$, where} \\
& & u_1'=\delta_{A_1}(u_1,a),\\
& & u_2' = \begin{cases} \delta_{A_2}(u_2,a)\cup \{0\}, & \text{if $u_{1}'\in F_{1}$}, \\
                         \delta_{A_2}(u_2,a),           & \text{otherwise},
           \end{cases}\\
& & u_i' = \begin{cases} \delta_{A_i}(u_i,a)\cup \{0\}, & \text{if $u_{i-1}'\cap F_{i-1} \neq \emptyset$}, \\
                         \delta_{A_i}(u_i,a),           & \text{otherwise},
           \end{cases} \quad \text{for $3\leq i\leq k$}
\end{eqnarray*}
%$\delta_E : \delta_E( \langle u_1, u_2, \ldots ,u_k \rangle ,a) =
%\langle u_1',u_2', \ldots ,u_k' \rangle$, for $a\in \Sigma$, where
%$u_1'=\delta_{A_1}(u_1,a)$; $u_2' = \delta_{A_2}(u_2,a)\cup \{0\}$
%if $u_{1}'\in F_{1}$, $u_2' = \delta_{A_2}(u_2,a)$ otherwise;
%$u_i' = \delta_{A_i}(u_i,a)\cup \{0\}$ if $u_{i-1}'\cap F_{i-1}
%\neq \emptyset$, $u_i' = \delta_{A_i}(u_i,a)$ otherwise; $3\leq
%i\leq k$.
where{\footnotesize
\begin{eqnarray*}
D' & = & Q_{1}\times \{\emptyset\} \times  ( 2^{Q_{3}}\times \cdots \times 2^{Q_{k}}-\{\emptyset\}^{k-2})
   {}+Q_{1}\times (2^{Q_{2}}-\{\emptyset\}) \times \{\emptyset\} \times (2^{Q_{4}}\times \cdots \times 2^{Q_{k}}-\{\emptyset\}^{k-3})+\cdots \\
   & & {}+Q_{1}\times (2^{Q_{2}}-\{\emptyset\}) \times \cdots \times (2^{Q_{k-2}}-\{\emptyset\})\times \{\emptyset\}\times (2^{Q_{k}}-\{\emptyset\});\\
E_1' & = & F_1\times (\{\emptyset\}^{k-1}\cup (2^{Q_{2}-\{0\}}-\{\emptyset\})\times (\{\emptyset\}^{k-2}\cup (2^{Q_{3}}-\{\emptyset\})\times \cdots 
\times (\{\emptyset\}^2\cup (2^{Q_{k-1}}-\{\emptyset\})2^{Q_{k}})\ldots ));\\
E_2' & = & (Q_{1}-F_1)\times ((2^{F_2}-\{\emptyset\})\cup 2^{Q_{2}-F_2})\times (\{\emptyset\}^{k-2}\cup (2^{Q_{3}-\{0\}}-\{\emptyset\})\times \cdots 
{}\times (\{\emptyset\}^2\cup (2^{Q_{k-1}}-\{\emptyset\})2^{Q_{k}})\ldots ))\\
   & & {}+F_1\times ((2^{F_2}-\{\emptyset\})\cup 2^{Q_{2}-F_2-\{0\}})\times (\{\emptyset\}^{k-2}\cup (2^{Q_{3}-\{0\}}-\{\emptyset\})\times \cdots 
   {}\times (\{\emptyset\}^2\cup (2^{Q_{k-1}}-\{\emptyset\})2^{Q_{k}})\ldots ));\\
%E_3' & = &  (Q_{1}-F_1)\times (2^{Q_{2}-F_2}-\{\emptyset\})\times ((2^{F_3}-\{\emptyset\})\cup 2^{Q_{3}-F_3})\times (\{\emptyset\}^{k-3}\cup (2^{Q_{4}-\{0\}}-\{\emptyset\})\\
%                & & \times \cdots \times (\{\emptyset\}^2\cup (2^{Q_{k-1}}-\{\emptyset\})2^{Q_{k}})\ldots ))\\
%                    & & +(Q_{1}-F_1)\times ((2^{F_2}-\{\emptyset\})\cup 2^{Q_{2}-F_2})\times ((2^{F_3}-\{\emptyset\})\cup 2^{Q_{3}-F_3-\{0\}})\times (\{\emptyset\}^{k-3}\\
%                        & & \cup  (2^{Q_{4}-\{0\}}-\{\emptyset\}) \times \cdots \times (\{\emptyset\}^2\cup (2^{Q_{k-1}}-\{\emptyset\})2^{Q_{k}})\ldots ))\\
%                            & & +F_1\times 2^{Q_{2}-F_2-\{0\}}\times ((2^{F_3}-\{\emptyset\})\cup 2^{Q_{3}-F_3})\times (\{\emptyset\}^{k-3}\cup (2^{Q_{4}-\{0\}}-\{\emptyset\}) \\
%                                & & \times \cdots \times (\{\emptyset\}^2\cup (2^{Q_{k-1}}-\{\emptyset\})2^{Q_{k}})\ldots ))\\
%                                    & & +F_1\times ((2^{F_2}-\{\emptyset\})\cup 2^{Q_{2}-F_2-\{0\}})\times ((2^{F_3}-\{\emptyset\})\cup 2^{Q_{3}-F_3-\{ 0 \}})\times (\{\emptyset\}^{k-3}\\
%                                        & & \cup (2^{Q_{4}-\{0\}}-\{\emptyset\}) \times \cdots \times (\{\emptyset\}^2\cup (2^{Q_{k-1}}-\{\emptyset\})2^{Q_{k}})\ldots ));\\
 \ldots \\
 %%%%%%%%%%%%%%%%%%%%%%%%%%%%%%E_i' has been added here%%%%%%%%%%%%%%%%%%%%%%%%%%%%%%
 E_i' & = &  (Q_{1}-F_1)\times (2^{Q_{2}-F_2}-\{\emptyset\})\times  \cdots \times ((2^{F_i}-\{\emptyset\})\cup 2^{Q_{i}-F_i})\\
   & & \times (\{\emptyset\}^{k-i}\cup (2^{Q_{i+1}-\{0\}}-\{\emptyset\}) \times \cdots \times (\{\emptyset\}^2\cup (2^{Q_{k-1}}-\{\emptyset\})2^{Q_{k}})\ldots ))\\
   & & {}+\cdots 
   {}+F_1\times ((2^{F_2}-\{\emptyset\})\cup 2^{Q_{2}-F_2-\{0\}})\times \cdots \times ((2^{F_i}-\{\emptyset\})\cup 2^{Q_{i}-F_i-\{ 0 \}})\times (\{\emptyset\}^{k-i}\\
   & & {}\cup (2^{Q_{i+1}-\{0\}}-\{\emptyset\}) \times \cdots \times (\{\emptyset\}^2\cup (2^{Q_{k-1}}-\{\emptyset\})2^{Q_{k}})\ldots )).%;\\
% \ldots \\
%E_{k-1}' & = &  (Q_{1}-F_1)\times (2^{Q_{2}-F_2}-\{\emptyset\})\times (2^{Q_{3}-F_3}-\{\emptyset\})\times \cdots \times (2^{Q_{k-2}-F_{k-2}}-\{\emptyset\}) \\
%                                    & &  \times ((2^{F_{k-1}}-\{\emptyset\})\cup 2^{Q_{k-1}-F_{k-1}}) \times 2^{Q_{k}-\{ 0 \}}\\
%                                      & & +\ldots \\
%                                        & & +F_1\times ((2^{F_2}-\{\emptyset\})\cup 2^{Q_{2}-F_2-\{0\}})\times ((2^{F_3}-\{\emptyset\})\cup 2^{Q_{3}-F_3-\{0\}})\times \cdots \\
%                                          & & \times ((2^{F_{k-2}}-\{\emptyset\})\cup 2^{Q_{k-2}-F_{k-2}-\{0\}})  \times ((2^{F_{k-1}}-\{\emptyset\})\cup 2^{Q_{k-1}-F_{k-1}-\{ 0 \}}) \times 2^{Q_{k}-\{ 0 \}}.
\end{eqnarray*}}

Intuitively, $Q_E$ is a set of $k$-tuples whose first component is
a state in $Q_1$ and the $i$th component is a subset of states in
$Q_i$, $2\leq i\leq k$.

$Q_E$ does not contain those $k$-tuples whose $i$th component is
$\emptyset$ and whose $j$th component is not $\emptyset$, when $1<
i<j \leq k$. $D'$ is the set of them.

$Q_E$ does not contain those $k$-tuples whose first component is
an element of $F_1$ and whose second component is not $\emptyset$
(if it is $\emptyset$ then all the elements afterward have to be
$\emptyset$) and does not contain $0$, either. $E_1'$ is the set
of them.

$Q_E$ does not contain those $k$-tuples whose $i$th component
contains one or more final states of DFA $A_i$ and whose $(i+1)$th
component is not $\emptyset$ (if it is $\emptyset$ then all the
elements afterward have to be $\emptyset$) and does not contain
$0$, when $2\leq i \leq k-1$, either. $E_i'$ is the set of them.

Clearly, $L(E)=L(A_1)\cdots L(A_k)$. Let $|Q_{A_i}|=n_i$ and
$|F_{A_i}|=f_i$, $1\leq i\leq k$.\\
Then $E$ has 
$n_12^{n_2+\dots+n_k}-D-E_1-E_2-\dots - E_{k-1}$ states. \end{proof}

%Clearly, $L(E)=L(A_1)\cdots L(A_k)$. Let $|Q_{A_i}|=n_i$ and
%$|F_{A_i}|=f_i$, $1\leq i\leq k$. Then $E$ has the following
%number of states:
%$$ n_12^{n_2+\ldots +n_k}-D-\sum_{i = 1}^{k-1} E_i $$\end{proof}
%\begin{eqnarray*}
% & & n_12^{n_2+\ldots +n_k}-D-\sum_{i = 1}^{k-1} E_i.
%\end{eqnarray*}
%\end{proof}

Note that when each $A_i$, $1\leq i\leq k$, has one final state,
this upper bound is exactly the same as the lower bound stated in
Theorem~\ref{lowerbound}. Thus, this bound is tight and is the
state complexity of the catenation of $k$ regular languages.

Although we have proved that this state complexity is tight, it is
too long and complex to be intuitive and comprehensible. Let
$SC_{CAT}(n_1, \ldots, n_k)$ denote the state complexity of
catenation of $k$ languages accepted by $n_1$-state, $\ldots$,
$n_k$-state DFAs, respectively, $n_1, \ldots, n_k \geq 2$. By
observing the structure of the result, we can see that $n_1
2^{n_2+\cdots+n_k}$ is a good approximation with the ratio bound
$$\frac{n_1 2^{n_2+\cdots+n_k}}{SC_{CAT}(n_1, \ldots, n_k)}< 4.$$
However, all our experiments show that the ratio bound for this
approximation is less than $3$, but we have not been able to prove
it.

\section{Conclusion}

The new concept of state complexity approximation is introduced.
It further advances the idea of state complexity estimation by
including the ratio bound. The ratio bound gives a precise and
intuitive measurement on the ``quality'' of the estimation.

We show that state complexity approximation can play useful roles
in two different cases. In the first case, the exact state
complexities have not been obtained. They may be very difficult to
obtain. However, approximation results with low ratio bounds can
be obtained rather easily and they are good enough for practical
purposes in general. In the second case, the exact state
complexities have been proved. The approximations of those results
with low ratio bounds can simplify the formulae of the
complexities and make them more intuitive and easier to apply.

Clearly, the state complexity approximation is a useful and
important concept. We expect many new results on state complexity
approximation will come out in the near future.

\bibliographystyle{eptcs}
\bibliography{gao}

\end{document}